\pdfoutput=1

\documentclass{article}

\usepackage{hyperref}
\usepackage[a4paper]{geometry}
\usepackage[utf8]{inputenc}
\usepackage[T1]{fontenc}
\usepackage{lmodern}
\usepackage{microtype}
\usepackage{graphicx}
\usepackage{amsmath,amssymb,amsthm}
\usepackage[dvipsnames]{xcolor}

\definecolor{cb-yellow}{RGB}{221,170,51}
\definecolor{cb-red} {RGB}{187,85,102}
\definecolor{cb-teal}{RGB}{0,153,136}
\definecolor{cb-blue} {RGB}{0,68,136}
\definecolor{cb-green}{RGB}{17,119,51}
\definecolor{cb-purple} {RGB}{170,68,153}
\definecolor{cb-palegrey} {RGB}{221,221,221}

\hypersetup{
  pdfauthor={Michael Borinsky, Shiyue Ren, Maximilian Wiesmann},
  pdftitle={An edge-bicolored graph approach to the Ising model on random regular graphs},
  pdfsubject={},
  pdfkeywords={},
    colorlinks=true,
    linkcolor=cb-red,
    filecolor=magenta,      
    urlcolor=cb-purple,
    citecolor=cb-yellow
}

\newtheorem{theorem}{Theorem}
\numberwithin{theorem}{section}
\newtheorem{corollary}[theorem]{Corollary}
\newtheorem{proposition}[theorem]{Proposition}
\newtheorem{lemma}[theorem]{Lemma}

\newtheorem{notation}[theorem]{Notation}

\theoremstyle{remark}
\newtheorem{remark}[theorem]{Remark}
\newtheorem{examplex}[theorem]{Example}
\newenvironment{example}
{\pushQED{\qed}\examplex}
{\popQED\endexamplex}

\newcommand{\RR}{\mathbb{R}}
\newcommand{\ZZ}{\mathbb{Z}}

\newcommand{\degv}{\mathrm{deg}} 
\DeclareMathOperator{\Aut}{Aut}

\newcommand{\GG}{\mathcal{G}}

\renewcommand{\O}{\mathcal{O}}

\usepackage[backend=biber,style=alphabetic,maxnames=99,url=false,isbn=false]{biblatex}
\DefineBibliographyStrings{english}{
  in = {}
}
\DeclareSourcemap{
  \maps[datatype=bibtex]{
    \map{
      \step[fieldsource=doi,
            match=\regexp{https?://(dx.)?doi.org/},
            replace={}]
    }
  }
}
\AtBeginBibliography{\small}

\title{\bf An edge-bicolored graph approach to the\\Ising model on random regular graphs}
\date{}
\author{Michael Borinsky \and Shiyue Ren \and Maximilian Wiesmann}

\begin{document}

\maketitle

\begin{abstract}
We give an exact solution of the ferromagnetic Ising model on a random regular graph ensemble via analytic combinatorics. Expressing the partition function as the generating function of labeled edge-bicolored graphs, we obtain the free energy in the thermodynamic limit from the asymptotic enumeration of these graphs. A simple analysis of the resulting formula reveals a second-order phase transition with critical exponents of the mean-field universality class.
\end{abstract}

\section{Introduction}
\label{sec:intro}

The Ising model \cite{lenz1920Ising} is a prominent model for ferromagnetism. Despite its simplicity, exact solutions (i.e.\ closed form expressions for the free energy function in the thermodynamic limit) have been obtained only for very few graph  families. Notable among these are Onsager's famous solution of the $\ZZ^2$ lattice Ising model \cite{OnsagerSolution}, and the solution of the ``infinite-dimensional'' mean-field model, see e.g.\ \cite[\S 3]{baxter1982exactly}.\par

Instead of tying the model to a single fixed graph, a different approach is to study the Ising model on a \emph{(dynamical) random graph}, i.e.\ 
to average the partition function over a graph ensemble. This approach has prominently been pursued by Kazakov \cite{KAZAKOV_Ising} and Boulatov--Kazakov \cite{BOULATOV_Kazakov_Ising}. They provide an exact solution for the Ising model on random regular planar graphs and compute its critical exponents for regularity three and four, by relating the random Ising model partition function to a matrix integral.\par

In this paper we study the Ising model on a (dynamical) \emph{random regular graph} with no planarity assumption. Since random regular graphs are \emph{locally tree-like}, this is closely related to the works \cite{Dembo_Montanari_Locally_Tree_Like,dembo2013factor,dommers2014ising}, where the Ising model is studied on locally tree-like graphs. 
These works rigorously confirm the predictions in \cite{Dorogovtsev:2002ix,leone2002ferromagnetic}, and establish that the phase transition coincides with that of the \emph{Bethe lattice model}, i.e.\ the Ising model on the ``interior'' of the infinite regular Cayley tree, see \cite{Chen1974Hopping} and also \cite[\S 4]{baxter1982exactly}. The critical exponents are found to be in the mean field universality class.\par

The articles mentioned above use probabilistic methods. In contrast, the present article uses techniques from \emph{analytic combinatorics} to solve the Ising model on random regular graphs. 
Our computation of the free energy function in the thermodynamic limit, i.e.\ the limit where the number of vertices becomes large, Theorem~\ref{thm:free_energy}, is based on a formula for the asymptotic number of edge-bicolored regular graphs with certain vertex incidence constraints, obtained by two of the present authors with Meroni \cite{BMW_bicolored,BMW_multicolor}. Concretely, we study the partition function
\begin{equation}
  \label{eq:random_partition_fct}
  Z_{n,k}(\beta, h) = \sum_{G \in \GG_{n,k}} \frac{1}{|\Aut(G)|} Z_G(\beta,h),
\end{equation}
where $\GG_{n,k}$ is the set of isomorphism classes of $k$-regular graphs $G$ with $n$ vertices and self-edges and multiple edges allowed. For a single graph $G$, 
\begin{align}
\label{eq:ZG}
  Z_G(\beta,h) = \sum_{\sigma\in \{\pm 1\}^{|V_G|}} e^{-\beta H_G(\sigma)} \quad \text{with} \quad H_G(\sigma) = -J\!\!\! \sum_{\{i,j\}\in E_G} \sigma_i\sigma_j - h \sum_{i \in V_G} \sigma_i
\end{align}
is the usual partition function of the Ising model on the graph $G$. Here, $\beta = (k_B T)^{-1}$ is the inverse thermodynamic temperature and $J \in \RR_{> 0}$ is a ferromagnetic coupling parameter that we fix once and for all throughout this article.
The automorphism group $\Aut(G)$ in \eqref{eq:random_partition_fct} may contain automorphisms that flip self-edges or permute multiple edges between one pair of vertices. 
This convention corresponds to choosing a random half-edge--labeled graph, i.e.~a graph in which each half-edge carries a unique label.\par

We define the \emph{free energy function} in the \emph{thermodynamic limit} as 

\begin{equation} 
  \label{eq:def_free_energy}
  f(\beta,h) := \frac{1}{\beta}
  \begin{cases}
    \lim_{n\to\infty} \frac{1}{n} \log \frac{Z_{n,k}(\beta, h)}{Z_{n,k}(0, 0)} & \text{ if $k$ is even,} \\
    \lim_{n\to\infty} \frac{1}{2n} \log \frac{Z_{2n,k}(\beta, h)}{Z_{2n,k}(0, 0)} & \text{ if $k$ is odd. }
  \end{cases}
\end{equation}

Distinguishing the two cases is necessary here, since $Z_{n,k}(\beta, h) =0$, whenever both $n$ and $k$ are odd, because finite odd-regular graphs exist only on an even number of vertices. The normalization by $Z_{n,k}(0,0)$ accounts for the growing number of graphs and makes $f(\beta,h)$ finite. The function defined in \eqref{eq:def_free_energy} is an \emph{annealed free energy function}, see e.g.\ \cite{FixedMagnetization} for a recent treatment. This differs from the \emph{quenched} random graph model, in which one first samples the graph and then studies the Ising model on that fixed graph.\par  

We prove that the Ising model on random $k$-regular graphs admits a second-order phase transition at the critical temperature $T_c = 2J \big( k_B \log\big(\tfrac{k}{k-2}\big) \big)^{-1}.$
This is the critical temperature of the Bethe lattice model, i.e.\ the Ising model on the ``interior'' of a Cayley tree with regularity $k$ \cite{Chen1974Hopping}, see also \cite[\S 4]{baxter1982exactly}. Such behavior has also been predicted for locally tree-like models in \cite{dommers2014ising}.
Our main result is the following expression for the free energy function.

\begin{theorem}
  \label{thm:free_energy}
  Let $k\geq 3$. The free energy of the random $k$-regular graph Ising model is
  \begin{equation}
    \label{eq:free_energy_theorem}
    f(\beta,h) = \frac{1}{\beta} \left( \log(\cosh(\beta h))+ \frac{k}{2}\log(\cosh(\beta J)) +  \log(k! \cdot V_k(\theta^\star(\beta,h))) \right).
  \end{equation}
  Here, with 
  $\phi = \arctan\left( \sqrt{\tanh(\beta J)} \right) \in (0,\pi/4)$,
$V_k(\theta)$ is the function
  \begin{equation*}
    V_k(\theta) = \frac{(1+\tanh(\beta J))^{k/2}}{2 \cdot k!}\left( (1+\tanh(\beta h))\cos^k(\theta - \phi) + (1-\tanh(\beta h))\cos^k(\theta + \phi) \right),
  \end{equation*}
and $\theta^\star(\beta,h)$ is the location of a global maximum of $|V_k(\theta)|$.

Let $\beta_c^{\mathrm{tree}} = \frac{1}{k_B T_c}= \frac{1}{2J}\log(\frac{k}{k-2})$.
The function $|V_k(\theta)|$ has exactly one global maximum
\begin{enumerate}
\item in the open interval $(0,\frac{\pi}{2}-\phi)$ if $h > 0$.
\item in the open interval $(-\frac{\pi}{2}+\phi,0)$ if $h < 0$.
\item in the open interval $(0,\frac{\pi}{2}-\phi)$ if $h =0$ and $\beta > \beta_c^{\mathrm{tree}}$.
\item at $\theta = 0$ if $h=0$ and $\beta < \beta_c^{\mathrm{tree}}$.
\end{enumerate}
At $(\beta, h) = (\beta_c^{\rm tree}, 0)$, $f$ can be completed via \eqref{eq:free_energy_theorem} to be continuous and differentiable at that point.
\end{theorem}

The proof of Theorem~\ref{thm:free_energy} relies on the graphical expansion of the Ising model, a reinterpretation of $Z_{n,k}$ as an enumeration of edge-bicolored graphs (see Section~\ref{sec:combinatorics}), and a formula for the asymptotic number of such graphs. Having obtained the expression~\eqref{eq:free_energy_theorem} for the free energy, we can determine the \emph{phase transitions} by simply analyzing which saddle points of $V_k(\theta)$ are global maxima. Here, a phase transition means that there exists a point $(\beta_c,h_c)$ of non-analyticity of the free energy function in the thermodynamic limit.
For nonzero external magnetic field $h \neq 0$, we find that the model does not exhibit any phase transitions. If, on the other hand, $h = 0$, the model has a second-order phase transition at $\beta_c^{\rm tree}$, as well as a first-order phase transition as $h$ passes over the ray $\{(\beta,0) \,:\, \beta > \beta_c^{\rm tree}\}$.

\begin{corollary}
  \label{cor:phase_transitions}
  Let $k\geq 3$ be an integer. The Ising model on a random \mbox{$k$-regular} graph has a second-order phase transition at $(\beta_c, h_c) = (\beta_c^{\mathrm{tree}}, 0)$ where $\beta_c^{\mathrm{tree}} = \frac{1}{2J}\log\big(\frac{k}{k-2}\big)$. That means $f(\beta,h)$ is once but not twice differentiable in $h$ at that point. 
  Furthermore, for $h=0$ and $\beta > \beta_c^{\rm tree}$, the free energy $f(\beta, h)$ is continuous, but not differentiable in $h$. Beyond these singularities, $f(\beta,h)$ is analytic for all other values of $h \in \RR$ and $\beta > 0$.
\end{corollary}

Further straightforward analysis lets us obtain the \emph{critical exponents} of our model at the phase transition. In the following we define various physical quantities associated with the Ising model and their critical exponents. See \cite[\S 1]{baxter1982exactly} for more details. The \emph{magnetization} is defined by
\begin{equation}
  \label{eq:def_mag}
  m(\beta,h) := \frac{\partial}{\partial h} f(\beta,h).
\end{equation}
By Theorem~\ref{thm:free_energy}, this derivative is well-defined almost everywhere, except for a possible discontinuity at $h=0$. The \emph{spontaneous magnetization} is the one-sided limit
\[
  m_0(\beta) := \lim_{h \searrow 0} m(\beta,h).
\]
Since $m$ is an odd function of $h$, if $m_0\neq 0$, then $m$ has a discontinuity in $h$. As stated in Corollary~\ref{cor:phase_transitions}, a phase transition at $(\beta,h) = (\beta_c,0)$ is \emph{second-order} if $m$ is continuous at $ (\beta_c,0)$, but $f$ is not twice differentiable in $h$ at this point. In that case, we may define the \emph{critical exponent $\beta^\mathrm{mag}$} of that phase transition. If we define the normalized temperature as $t := (T -T_c)/T_c$ and the spontaneous magnetization behaves asymptotically like
\begin{equation}
  \label{eq:def_beta^{mag}}
  m_0(t) \sim (-t)^{\beta^\mathrm{mag}} \quad \text{as } t \nearrow 0,
\end{equation}
then 
$\beta^\mathrm{mag}$ is the critical magnetization exponent of the model at $(\beta_c^\mathrm{tree},0)$. Here, the notation $f(t) \sim g(t) \text{ as } t \nearrow 0$ means that the limit $\lim_{t\nearrow 0} f(t)/g(t)$ exists and is nonzero. \par
Moreover, if the magnetization scales near the phase transition point as
\[
  m(\beta_c^{\rm tree},h) \sim h^{1/\delta} \quad \text{as } h \searrow 0,
\]
then $\delta$ is a critical exponent at $(\beta_c^{\rm tree}, 0)$. The second derivative of the free energy is the \emph{magnetic susceptibility}
\begin{equation*}
  \chi(\beta,h) := \frac{\partial^2}{\partial h^2} f(\beta, h).
\end{equation*}
If, at critical temperature and zero magnetic field, the susceptibility scales as 
\[
  \chi(t, 0) \sim t^{-\gamma} \quad \text{as } t \searrow 0,
\]
$\gamma$ is another critical exponent at $(\beta_c^{\rm tree}, 0)$. Finally, we can define the specific heat critical exponent $\alpha$ by comparing the analytic continuation of the high-temperature free energy function with the low-temperature free energy. Define $f_{\rm ht}(t) := \frac{k}{2\beta(t)} \log(\cosh(\beta(t) J))$. If the limit
\begin{equation}
  \label{eq:singular_free_energy}
  \lim_{t \searrow 0} \frac{\log \left| f(-t,0) - f_{\rm ht}(-t) \right|}{\log t}
\end{equation}
is well-defined and equals $2-\alpha$, then $\alpha$ is the specific heat critical exponent. In particular, $2-\alpha$ should be thought of as the order of the phase transition. As a second consequence of our main result, Theorem~\ref{thm:free_energy}, we obtain the following characterizations of the phase transitions.

\begin{corollary}
 \label{cor:critical_exponents}
    Let $k\geq 3$ be an integer. For $h=0$ and $\beta > \beta_c^{\rm tree}$ the random $k$-regular graph Ising model has a nonzero spontaneous magnetization given by
    \begin{equation}
    m_0(\beta) = \frac{(1+u)^k - (1-u)^k}{(1+u)^k + (1-u)^k},
  \end{equation}
  where $u$ is the unique positive solution to the equation
  \begin{equation}
    \frac{\tanh(\beta J) - u}{\tanh(\beta J) + u} = \left( \frac{1-u}{1+u} \right)^{k-1}.
  \end{equation}
  The second-order phase transition at $(\beta_c^{\rm tree}, 0)$ is characterized by the critical exponents
  \begin{equation}
    \label{eq:crit_exp}
    \alpha=0,\quad \beta^{\rm mag} = \frac{1}{2},\quad \gamma=1 \quad \text{and} \quad \delta = 3.
  \end{equation}
\end{corollary}

The nonzero spontaneous magnetization for $h=0$ and $\beta > \beta_c^{\rm tree}$ is explained by the spontaneous breaking of the global $\mathbb{Z}/2\mathbb{Z}$ spin-flip symmetry.
The critical exponents \eqref{eq:crit_exp} coincide with those of the mean-field model \cite[\S 3]{baxter1982exactly}. This agrees with the expectation that our model belongs to the ``infinite-dimensional'' mean field universality class, since random regular graphs are \emph{locally tree-like} \cite{Dembo_Montanari_Locally_Tree_Like}.\par

\paragraph{Outline of the article.} In Section~\ref{sec:combinatorics}, we recall the graphical expansion of the Ising model to connect the partition function $Z_G$ for a fixed graph $G$ to a sum over certain subgraphs of $G$. The average partition function $Z_{n,k}$ over all $k$-regular graphs is considered in Section~\ref{sec:limit_free_energy} where we relate it to the enumeration of specific edge-bicolored graphs. By studying the asymptotics of these graphs we prove Theorem~\ref{thm:free_energy}. In Section~\ref{sec:crit_exponents}, we compute the critical exponents and prove Corollaries~\ref{cor:phase_transitions} and \ref{cor:critical_exponents}. In the final Section~\ref{sec:magnetization_combinatorial}, we provide a purely combinatorial interpretation of the magnetization, again as an enumeration of edge-bicolored graphs.

\paragraph{Acknowledgments.}
Research at Perimeter Institute is supported in part by the Government of Canada through the Department of Innovation, Science and Economic Development Canada and by the Province of Ontario through the Ministry of Colleges and Universities.

\section{A combinatorial approach to the Ising model}
\label{sec:combinatorics}

The partition function $Z_G(\beta,h)$ of the Ising model on an individual graph $G$ can be written as a sum over certain subgraphs of $G$. This is sometimes called the high-temperature expansion of the Ising model and is originally due to van der Waerden \cite{vanderWaerden1941}. In this section, we briefly explain how the expansion works to make the paper self-contained. See \cite{duminil2016random} for a modern mathematical and more detailed treatment of expansions of the Ising model partition function.\par

By the definition \eqref{eq:ZG}, we have 
\begin{align*}
  Z_G(\beta,h) = 
\sum_{\sigma\in \{\pm 1\}^{|V_G|}} 
\left(
\prod_{\{i, j\} \in E_G} e^{\beta J \sigma_i \sigma_j} \right) \left( \prod_{v \in V_G} e^{\beta h \sigma_{v}} \right)
\end{align*}

It is convenient to make the following changes of variables.

\begin{notation}
\label{not:variable_changes}
We define $\kappa = \tanh(\beta J)$, $\tau = \tanh(\beta h)$, $\rho = \cosh(\beta J)$ and $\eta = \cosh(\beta h)$. Note that high temperature $T > T_c$ means $\kappa < \kappa_c = \tanh\big(\tfrac{1}{2} \log \big(\tfrac{k}{k-2}\big)\big) = \tfrac{1}{k-1}$.
\end{notation}

It follows from the identity $e^{K \sigma} = \cosh(K) (1 + \tanh(K) \cdot \sigma)$, which holds for $\sigma \in \{+1, -1\}$ and any $K \in \RR$ that 
$
    e^{\beta J \sigma_v \sigma_u} = \rho (1 + \kappa \sigma_v \sigma_u)$  and 
    $e^{\beta h \sigma_v} = \eta (1 + \tau \sigma_v).
$
Hence, 
\begin{equation}
\label{eq:ZGkappaeta}
    Z_G(\beta,h) = \rho^{|E_G|} \eta^{|V_G|} 
\sum_{\sigma\in \{\pm 1\}^{|V_G|}} 
\left( \prod_{\{i, j\} \in E_G} (1 + \kappa \sigma_i \sigma_j) \right) \left( \prod_{v \in V_G} (1 + \tau \sigma_{v}) \right).
\end{equation}
Expanding the products over $E_G$ and $V_G$ gives a sum over all subsets of $E_G$ and $V_G$. So,
\begin{equation}
\label{eq:ZGlong}
    Z_G(\beta,h) = \rho^{|E_G|} \eta^{|V_G|} 
\sum_{\gamma \subseteq E_G} 
\kappa^{|\gamma|}
\sum_{\mu \subseteq V_G}
\tau^{|\mu|}
\sum_{\sigma\in \{\pm 1\}^{|V_G|}} 
\left( \prod_{\{i, j\} \in \gamma} \sigma_i \sigma_j \right) \left( \prod_{v \in \mu} \sigma_{v} \right).
\end{equation}
\begin{lemma}
\label{lmm:evenodd}
Let $(\gamma, \mu)$ be a pair $\gamma \subseteq E_G$ and $\mu \subseteq V_G$.
We say $(\gamma, \mu)$ is matched if $\mu$ contains exactly those vertices which have an odd number of incident edges in $\gamma$.
We have the identity
\[
\sum_{\sigma\in \{\pm 1\}^{|V_G|}} 
\left( \prod_{\{i, j\} \in \gamma} \sigma_i \sigma_j \right) \left( \prod_{v \in \mu} \sigma_{v} \right)
=
\begin{cases}
2^{|V_G|} & \text{ if $(\gamma, \mu)$ is matched} \\
0 & \text{otherwise}
\end{cases}
\]
\end{lemma}
\begin{proof}
If $(\gamma, \mu)$ is matched, the summands are all 1. Otherwise, 
there is some vertex $v$ such that $\sigma_v$ appears with an odd power in the sum.
By symmetry, the sum vanishes.
\end{proof}

Consequently, the partition function $Z_G(\beta,h)$ can be expressed as a sum over colorings of the graph $G$
weighted by a product over its vertices
\begin{proposition}
  \label{prop:combinatorial_formula_partition_fct}
A cancellation-free formula for the Ising model partition function is
\begin{equation*}
    Z_G(\beta,h) = \sum_{\gamma \subseteq E_G} \prod_{v\in V_G} \lambda_{v}^{(G,\gamma)},
\end{equation*}
where the effective vertex weights are defined with $\omega^{(G,\gamma)}_v = 2\eta\cdot \rho^{\degv_G(v)/2}\cdot \kappa^{\degv_\gamma(v)/2}$ as:

\begin{equation}
\label{eq:comb_vtx_weights}
    \lambda_{v}^{(G,\gamma)} = 
    \begin{cases}
        \omega^{(G,\gamma)}_v  & \text{if } \degv_\gamma(v) \text{ is even}, \\
        \omega^{(G,\gamma)}_v \cdot \tau & \text{if } \degv_\gamma(v) \text{ is odd},
    \end{cases}
\end{equation}
and 
where $\degv_G(v)$ is the degree of the vertex $v$ in the graph $G$ and 
$\degv_\gamma(v)$ is the degree of $v$ in the subgraph $\gamma$.
\end{proposition}
\begin{proof}
Use Lemma~\ref{lmm:evenodd} on \eqref{eq:ZGlong}.
\end{proof}

\begin{example}
Let $G$ be the $\theta$ graph with  two vertices, $v,u$ connected by 3 edges $e_1,e_2,e_3$.
We can compute the Ising model partition function directly using \eqref{eq:ZGkappaeta} to get
\begin{align*}
    Z_G(\beta,h) & = \rho^3 \eta^2 ( (1+\kappa)^3(1+\tau)^2 + (1+\kappa)^3(1-\tau)^2 + 2(1-\kappa)^3(1+\tau)(1-\tau)) \\
    &= 4\rho^3 \eta^2  (1 + 3\kappa^2 + 3\kappa\tau^2 + \kappa^3\tau^2).
\end{align*}
Applying Proposition~\ref{prop:combinatorial_formula_partition_fct} with the following $\lambda$-values for subgraphs of the $\theta$ graph
\begin{align*}
\lambda_{(v)}^{(G,\emptyset)} &=\lambda_{(u)}^{(G,\emptyset)}= 2  \eta \rho^{3/2}, &
\lambda_{(v)}^{(G,E_G)} &=\lambda_{(u)}^{(G,E_G)}= 2 \eta (\rho\kappa)^{3/2} \tau,\\
\lambda_{(v)}^{(G,\{e_k\})} &=\lambda_{(u)}^{(G,\{e_k\})}= 2 \eta 
\rho^{3/2} \kappa^{1/2}\tau, &
\lambda_{(v)}^{(G,E_G \setminus\{e_k\})} &=\lambda_{(u)}^{(G,E_G \setminus\{e_k\})}= 2 \eta 
\rho^{3/2} \kappa,
\end{align*}
we get
\begin{align*}
    Z_G(\beta,h) 
    &= \left(
    \lambda_{(v)}^{(G,\emptyset)}\lambda_{(u)}^{(G,\emptyset)} + \lambda_{(v)}^{(G,E_G)}\lambda_{(u)}^{(G,E_G)} \right. \\
    &\left.
    +\sum_{j=1}^3 
    \lambda_{(v)}^{(G,\{e_j\})}\lambda_{(u)}^{(G,\{e_j\})} +
    \sum_{j=1}^3 
    \lambda_{(v)}^{(G,E_G\setminus\{e_j\})}\lambda_{(u)}^{(G,E_G\setminus\{e_j\})} 
     \right) \\
    &= 4 \rho^3 \eta^2 (1 \cdot 1 + (\kappa^{3/2}\tau)^2 + 3(\kappa^{1/2}\tau)^2 + 3\kappa \cdot \kappa)  \\
    &= 4 \rho^3 \eta^2 (1 + \kappa^3\tau^2 + 3\kappa\tau^2 + 3\kappa^2)\, ,
\end{align*}
which matches the previous computation.
\end{example}

\begin{remark}
  Note that, combinatorially, the free energy function is a generating function of \emph{connected graphs}, by an application of the exponential formula, see e.g.\ \cite[\S 1.4]{bergeron1998combinatorial}.
\end{remark}

\section{Thermodynamic limit and phase transitions}
\label{sec:limit_free_energy}

In this section, we study the thermodynamic limit of the free energy function $f(\beta,h)$ stated in~\eqref{eq:def_free_energy} and prove Theorem~\ref{thm:free_energy}.
We find a closed form expression for $f(\beta,h)$ which is analytic for almost all values of $\beta$ and $h$. The points of non-analyticity are the locations of phase transitions. 
We also set up the proof of Corollary~\ref{cor:phase_transitions} in Corollary~\ref{cor:smoothf} of this section.

To evaluate the thermodynamic limit~\eqref{eq:def_free_energy}, we combine Proposition~\ref{prop:combinatorial_formula_partition_fct} with a result on the asymptotic enumeration of edge-bicolored graphs in \cite{BMW_bicolored}. Recall the changes of variables from Notation~\ref{not:variable_changes}.
We define the following  polynomial, which we call the \emph{potential},
\begin{equation}
  \label{eq:potential_def}
  V_k(x,y;\beta,h) =  \sum_{\ell=0}^k \kappa^{\ell/2} \tau^{\ell \bmod 2} \frac{x^{k-\ell}y^\ell}{\ell!(k-\ell)!},
\end{equation}
where $\ell \bmod 2 = 0$ if $\ell$ is even and $\ell \bmod 2 = 1$ otherwise.

\begin{proposition}
  \label{prop:asymptotic_limit_Z_n,k}
  Let $k\geq 3$. For fixed $\beta,h$, let $\Phi({\beta,h})$ be the set of global maxima of $|V_k(x,y)|= |V_k(x,y;\beta,h)|$, where $(x,y)$ ranges over  the unit circle $(x,y) \in S^1$. We define the following Hessian determinant-like term on the circle
  \[
    B_k(x,y) = k^2 - \frac{\partial_x^2 V_k(x,y) + \partial_y^2 V_k(x,y)}{V_k(x,y)}.
  \]
  Let $\ell = \frac{k-2}{2} n$. If $B_k(x,y) \neq 0$ for all  $(x,y) \in \Phi({\beta,h})$, then for $n\rightarrow \infty$,
  \[
    Z_{n,k}(\beta,h) \sim 
    \left\{\begin{array}{ll}
      (2\eta)^{n}\rho^{\frac{nk}{2}}\frac{1}{2\sqrt{2}\pi} k^{\frac{nk+ 1}{2}} \left( \frac{2}{k-2} \right)^{\ell-\frac{1}{2}} \Gamma(\ell) \sum_{(x,y) \in \Phi(\beta,h)} \frac{V_k(x,y)^{n}}{\sqrt{B_k(x,y)}} & \text{if } nk \text{ is even}, \\
        0 & \text{otherwise}\, 
    \end{array}\right.
  \]
\end{proposition}

\begin{proof}
  This is a consequence of Proposition~\ref{prop:combinatorial_formula_partition_fct}, the definition \eqref{eq:random_partition_fct} of $Z_{n,k}$ and Proposition~5.1 of \cite{BMW_bicolored}.
  To apply the latter, note that a pair $(G,\gamma)$ corresponds to an edge-bicolored graph, and translating the vertex weights in \eqref{eq:comb_vtx_weights} into weights of vertices with bicolored half-edges yields the polynomial in \eqref{eq:potential_def}. To apply Proposition~\ref{prop:combinatorial_formula_partition_fct}, use the orbit-stabilizer theorem to write \eqref{eq:random_partition_fct} as a sum over
  edge-bicolored graphs weighted by the inverse cardinality of their (color preserving) automorphism groups.
\end{proof}

\begin{corollary}
\label{cor:normalization}
Let $k\geq 3$ and $\ell = \frac{k-2}{2} n$, then for $n\rightarrow \infty$,
\[
  Z_{n,k} (0,0) \sim 
  \left\{\begin{array}{ll}
    2^{n}\frac{\sqrt{k-2}}{2\pi} \left( \frac{2k}{k-2} \right)^{\ell} \Gamma(\ell) ((k-1)!)^{-n} & \text{if } nk \text{ is even}, \\
    0 & \text{otherwise}\, 
  \end{array}\right.
\]
\end{corollary}

\begin{proof}
For $\beta = h = 0$, we have $\rho = \eta =1$ and $\tau = \kappa =0$. So, $V_k(x,y)= \frac{x^k}{k!}$ and $B_k(x,y) = k^2 - \frac{k(k-1)}{x^2}$.
The function $|V_k(x,y)|$ has exactly two maxima on the circle $(x,y) \in S^1$, so $\Phi(0,0) = \{ (1,0), (-1,0) \}$. As $B_k(x,y) \neq 0$ for both maxima, by Proposition~\ref{prop:asymptotic_limit_Z_n,k} we get
\[
  Z_{n,k} (0,0) \sim 
  \left\{\begin{array}{ll}
    2 \cdot 
    2^{n}\frac{1}{2\sqrt{2}\pi} k^{\frac{nk+ 1}{2}} \left( \frac{2}{k-2} \right)^{\ell-\frac{1}{2}} \Gamma(\ell) \frac{(k!)^{-n}}{\sqrt{k}} & \text{if } nk \text{ is even}, \\
    0 & \text{otherwise}
  \end{array}\right.
\]
Reordering terms gives the statement.
\end{proof}

To apply Proposition~\ref{prop:asymptotic_limit_Z_n,k} with nontrivial $\beta,h$ for the proof of Theorem~\ref{thm:free_energy}, we first prove some properties of the polynomial $V_k(x,y)$ from~\eqref{eq:potential_def}.
The following lemma is a direct consequence of the positivity of the coefficients of $V_k$.
\begin{lemma}
\label{lmm:positiveVk}
If $\kappa, \tau \geq 0$, then we have $|V_k(x,y)| \leq V_k(|x|, |y|)$ for all $(x,y) \in S^1$.
Moreover, if $\kappa \geq 0$ and $\tau \leq 0$, then $|V_k(x,y)|  \leq V_k(|x|,-|y|)$ for all $(x,y) \in S^1$. \qed
\end{lemma}
The polynomial $V_k(x,y)$ decomposes into a linear combination of two $k$-th powers,
\[
  V_k(x,y) = \frac{1+\tau}{2\cdot k!} (x + \sqrt{\kappa} y)^k + \frac{1-\tau}{2\cdot k!} (x - \sqrt{\kappa}y)^k\,,
\]
and parameterizing the circle as $x = \cos\theta$ and $y = \sin\theta$ and letting $\phi = \arctan(\sqrt{\kappa})$ gives
\begin{align*}
  V_k(\theta) = V_k(\cos(\theta),\sin(\theta)) = \frac{(1+\kappa)^{k/2}}{2\cdot k!}\left( (1+\tau)\cos^k(\theta - \phi) + (1-\tau)\cos^k(\theta + \phi) \right).
\end{align*}

To apply Proposition~\ref{prop:asymptotic_limit_Z_n,k}, we need to maximize $|V_k(\theta)|$ over the domain $\theta \in [-\pi,\pi]$.
By Lemma~\ref{lmm:positiveVk}, we may restrict the search for the global maximum to $\theta \in [0,\pi/2]$ if $h \geq 0$ and $\theta \in [-\pi/2,0]$ if $h \leq 0$.
Moreover, as $\beta > 0$, we also have $\kappa \in (0,1)$ and therefore $\phi=\arctan(\sqrt{\kappa}) \in (0,\pi/4)$.
We can write the critical equation, $V_k'(\theta^\star) = 0$, as
\begin{align}
  \label{eq:crit}
  g(\theta^\star) = \frac{1-\tau}{1+\tau}\, ,
\end{align}
where we defined the function $g:[0,\theta_s) \cup (\theta_s,\pi/2] \rightarrow \RR$ via
\begin{align}
  \label{eq:gdef}
  g(\theta) = \frac{\sin(\phi-\theta)}{\sin(\phi+\theta)} \left( \frac{\cos(\phi-\theta)}{\cos(\phi+\theta)} \right)^{k-1}\,,
\end{align}
which has a $(k-1)$-th order pole at $\theta= \theta_s = \pi/2-\phi \in (\pi/4,\pi/2)$.

\begin{lemma}\mbox{}
  \label{lem:maximum}
  \begin{enumerate}
    \item If $h > 0$, then $|V_k(\theta)|$ has a unique global maximum in the interval $[0,\pi/2]$ at $\theta= \theta^\star$.
    \item If $h < 0$, then $|V_k(\theta)|$ has a unique global maximum in the interval $[-\pi/2,0]$ at $\theta=-\theta^\star$.
  \end{enumerate}

  Furthermore, let $\beta_c^{\mathrm{tree}} = \frac{1}{2J}\log(\frac{k}{k-2})$.  If $h =0$ and
  \begin{enumerate}
    \item[3.] $\beta > \beta_c^{\mathrm{tree}}$, then $|V_k(\theta)|$ has a unique global maximum in the interval $[0,\pi/2]$ at $\theta = \theta^\star$.
    \item[4.] $\beta < \beta_c^{\mathrm{tree}}$, then $|V_k(\theta)|$ has a unique global maximum in the interval $[0,\pi/2]$ at $\theta = 0$.
  \end{enumerate}

  Here, $\theta^\star$ is the unique solution of 
  \begin{align}
    g(\theta^\star) = \frac{1-|\tau|}{1+|\tau|}\, ,
  \end{align}
  in the interval $(0,\pi/2-\phi)$ with $g$ defined in eq.~\eqref{eq:gdef}.
  (Recall that $\tau$ is a function of $h$ by Notation~\ref{not:variable_changes}.)
\end{lemma}

\begin{proof}
  The derivative of $g$ fulfills
  \begin{align*}
    g'(\theta) = \sin(2\phi)\frac{\cos(\phi-\theta)^{k-2}}{2\cos(\phi+\theta)^{k}\sin^2(\phi+\theta)} \bigl( (k-2) \cos(2\theta) - k \cos(2\phi) \bigr)\,,
  \end{align*}
  and the equation
  \begin{equation}
    \label{eq:theta_0}
    \cos(2\theta_0) = \frac{k}{k-2}\cos(2\phi)\, ,
  \end{equation}
  has a unique solution $\theta_0 \in [0,\pi/2]$ if and only if 
  $\cos(2\phi) \leq \frac{k-2}{k}$. If such a unique $\theta_0$ exists, then $g'(\theta_0)=0$. This condition is equivalent to $\kappa \geq \frac{1}{k-1}$ as $\cos(2\phi) = \frac{1-\kappa}{1+\kappa}$ 
  and also to $\beta \geq \frac{1}{J} \tanh^{-1}\left( \frac{1}{k-1} \right) = \beta_c^\mathrm{tree}$.

  If there is such a $\theta_0$, then $0<\theta_0 < \pi/4 < \theta_s$
  and $g$ is monotonically increasing in the interval $[0,\theta_0]$ and
  monotonically decreasing in the interval $[\theta_0,\theta_s)$. 
  On the other hand, we find that  if $\beta < \beta_c^\mathrm{tree}$, then
  $g$ is monotonically decreasing on $[0,\theta_s)$.

  On $(\theta_s,\pi/2]$, the derivative of $g$ is non-vanishing. Because $g$ has a pole at $\theta_s$, 
  its magnitude $|g(\theta)|$ must be decreasing on $(\theta_s,\pi/2]$.
  Because $|g(\pi/2)|=1$, it follows that $|g(\theta)| >1$ for $\theta \in (\theta_s,\pi/2)$.
  Furthermore, note that $g(0) =1$ and $\lim_{\theta \rightarrow \theta_s^-} g(\theta) = -\infty$.

  If $h > 0$, and therefore $0<\frac{1-\tau}{1+\tau} < 1$, \eqref{eq:crit} has a unique solution at $\theta^\star \in (0,\theta_s)$,
  as either $g$ is monotonically decreasing in that interval or $g$ is increasing from $g(0)=1$ to some maximal value $g(\theta_0) >1$ before it 
  decreases to $-\infty$ in the interval $[\theta_0,\theta_s)$. The case $h<0$ follows by the symmetry $(\tau,\theta) \leftrightarrow (-\tau,-\theta)$ of the potential $V_k$.
  There are no other critical points in the interval $[0,\pi/2]$. So, by Lemma~\ref{lmm:positiveVk}, the critical point of $V_k$ at $\theta^\star$ is a global maximum.

  This proves the $h \neq 0$ case of the statement.

  If $h=0$, \eqref{eq:crit} reduces to $g(\theta^\star) = 1$. 
  If, moreover, $\beta < \beta_c^\mathrm{tree}$, then \eqref{eq:crit} 
  has one or two solutions on the interval $[0,\pi/2]$ as $g(0) = 1$ and $g(\pi/2) = (-1)^k$.
  The discussion above implies that $g(\theta) \neq 1$ for all $\theta \in (0,\pi/2)$.
  It is easy to verify that $|V_k(0)| > |V_k(\pi/2)|$ as $\phi \in (0,\pi/4)$. Hence, only $\theta = 0$ is a
  global maximum of $V_k$ by Lemma~\ref{lmm:positiveVk}.

  If $\beta > \beta_c^\mathrm{tree}$, there is an additional candidate 
  for a global maximum at $\theta^\star \in (\theta_0,\theta_s)$, where \eqref{eq:crit} has a nontrivial solution.
  Computing the second derivative of $V_k$ and evaluating at $\theta=0$ gives
  \begin{equation}
      \label{eq:2nd_der_at_0}
      \frac{\partial^2}{\partial \theta^2} V_k(0) = \frac{2(1+\kappa)^{k/2}}{2 k!} k \cos^k(\phi) \left( (k-1) \tan^2(\phi) - 1 \right).
  \end{equation}
  Since $\beta > \beta_c^{\text{tree}} = \frac{1}{J}\mathrm{arctanh} \big(\tfrac{1}{k-1}\big)$, we have $\tan^2(\phi) = \kappa > \frac{1}{k-1}$. Moreover, $\cos^k(\phi) > 0$ as $0<\phi<\pi/4$, hence \eqref{eq:2nd_der_at_0} is positive and $\theta=0$ is a local minimum of $V_k$. Since $\theta=0$ and $\theta=\theta^\star$ are the only critical points, $\theta^\star$ is a local and hence global maximum.
\end{proof}

\begin{lemma}
  \label{lem:nondegenerate}
  In any of the four cases of Lemma~\ref{lem:maximum}, the maximum is nondegenerate.
\end{lemma}

\begin{proof}
  We first show that $\frac{\partial^2}{\partial \theta^2} V_k(\theta^\star) \neq 0$. Computing the derivative gives 
  \begin{align*}
    \frac{\partial^2}{\partial \theta^2} V_k(\theta) & = \frac{(1+\kappa)^{k/2}}{2 k!} \left( (1+\tau) k \cos^k(\phi - \theta)\left( (k-1)\tan^2(\phi - \theta) -1 \right) \right. \\ 
    & - \left. (1-\tau) k \cos^{k-2}(\phi + \theta) \left( \cos^2(\phi+\theta) - (k-1)\sin^2(\phi+\theta) \right) \right)
  \end{align*}
  Evaluating at $\theta = \theta^\star$, using \eqref{eq:crit}, and $\sin(\phi+\theta^\star) \neq 0 \neq \cos(\phi+\theta^\star)$, it remains to show
  \begin{equation}
    \label{eq:tan_neq_0}
    (k-1) \left( \tan(\phi-\theta^\star) + \tan(\phi+\theta^\star) \right) - \cot(\phi-\theta^\star) - \cot(\phi+\theta^\star) \neq 0.
  \end{equation}
  The left-hand side of \eqref{eq:tan_neq_0} is zero if and only if 
  \begin{enumerate}
    \item $\tan(\phi-\theta^\star) + \tan(\phi+\theta^\star) = 0$; or
    \item $\tan(\phi-\theta^\star) \tan(\phi+\theta^\star) = \frac{1}{k-1}$.
  \end{enumerate}
  The first case is equivalent to $\sin(2\phi) = 0$, which cannot occur since $\kappa > 0$. The second case is equivalent to $\cos(2\theta^\star) = \frac{k}{k-2} \cos(2\phi)$. This is precisely \eqref{eq:theta_0} from above, with unique solution $\theta_0$. But the proof of Lemma~\ref{lem:maximum} shows $\theta^\star\neq \theta_0$.\par

  Finally, consider case 4 of Lemma~\ref{lem:maximum}, i.e.\ $\tau = 0$, $\beta < \beta_c^{\text{tree}}$ and the global maximum is at $\theta=0$. The second derivative of $V_k$ at $\theta=0$ has already been computed in \eqref{eq:2nd_der_at_0}.
  Since $\phi \in (0,\pi/4)$, we have $\cos(\phi)\neq 0$. Moreover, $\tan^2(\phi) = \kappa < \frac{1}{k-1}$, since $\beta < \beta_c^{\text{tree}} = \frac{1}{J}\mathrm{arctanh} \big(\tfrac{1}{k-1}\big)$. Therefore, \eqref{eq:2nd_der_at_0} does not vanish.
\end{proof}

The values $(\beta,h)$ for which some saddle points of $V_k$ become degenerate form a \emph{discriminant} locus. Lemma~\ref{lem:nondegenerate} above shows that the discriminant is tightly connected to the phase transitions. Also see \cite[Remark 3.5]{LeeYang} for a discussion on the role of the discriminant in determining phase transitions and partition function zeros. 

\begin{proof}[Proof of Theorem~\ref{thm:free_energy}]
As a consequence of Lemma~\ref{lem:maximum} and \ref{lem:nondegenerate}, we may apply Proposition~\ref{prop:asymptotic_limit_Z_n,k} to
compute the large-$n$ behavior of $Z_{n,k}(\beta,h)$.
There may be multiple global maxima of $|V_k|$, but identifying one of them is sufficient, as their multiplicity only contributes an overall prefactor that does not contribute in the $n\rightarrow \infty$ limit after taking logarithms and normalizing by $n^{-1}$.
Computing the limit~\eqref{eq:def_free_energy} using Corollary~\ref{cor:normalization} gives the stated formula for $f(\beta,h)$.
\end{proof}

We conclude this section by proving Corollary~\ref{cor:phase_transitions}.

\begin{corollary}
  \label{cor:smoothf}
    The function $f(\beta,h)$ is analytic for all $(h,\beta)$ where $\beta > 0$
    with the exception of the ray $(0,\beta)$ where $\beta \geq \beta_c^{\rm tree}$.
\end{corollary}

\begin{proof}
  By Theorem~\ref{thm:free_energy}, the only source of non-analyticity of the free energy $f(\beta,h)$ is 
  the function $\theta^\star(\beta,h)$ that can behave non-analytically when the location of the global maximum of $|V_k|$ changes.
  It follows from Lemma~\ref{lem:maximum} that the value of $\theta^\star(\beta,h)$ is 
  determined by a unique maximum if $|h| > 0$. The location of this maximum can only change 
  analytically with $\beta,h$ and it is therefore sufficient to prove that $\theta^\star(\beta,h)$ is an analytic function whenever $0<\beta < \beta_c^\mathrm{tree}$.

  Recall that the location of critical points of $|V_k|$ is governed by equation~\eqref{eq:crit}.
  As explained in the proof of Lemma~\ref{lem:maximum}, the function $g(\theta)$ is monotonically decreasing in the interval $[0,\theta_s]$ if $\beta < \beta_c^\mathrm{tree}$.
  For small $h >0$, the location of the global maximum, as guaranteed by Lemma~\ref{lem:maximum}, approaches $\theta^\star \to 0$ smoothly. For $h < 0$, there is no more critical point with small $\theta^\star >0$.
  Hence, the function $\theta^\star(\beta,h)$ that tracks the global maximum 
  must be analytic in $\beta$ and $h$.
\end{proof}

\section{Magnetization and critical exponents}
\label{sec:crit_exponents}

In this section, we prove Corollary~\ref{cor:critical_exponents}. The critical exponents are obtained from our free energy formula in Theorem~\ref{thm:free_energy} by computing low-order Taylor expansions. We explain this in detail for the magnetization and the critical exponent $\beta^{\rm mag}$ in Subsection~\ref{subsec:crit_behavior}, and briefly present the analogous computations for the other critical exponents. 

\begin{figure}
  \centering
  \includegraphics[width=0.48\linewidth]{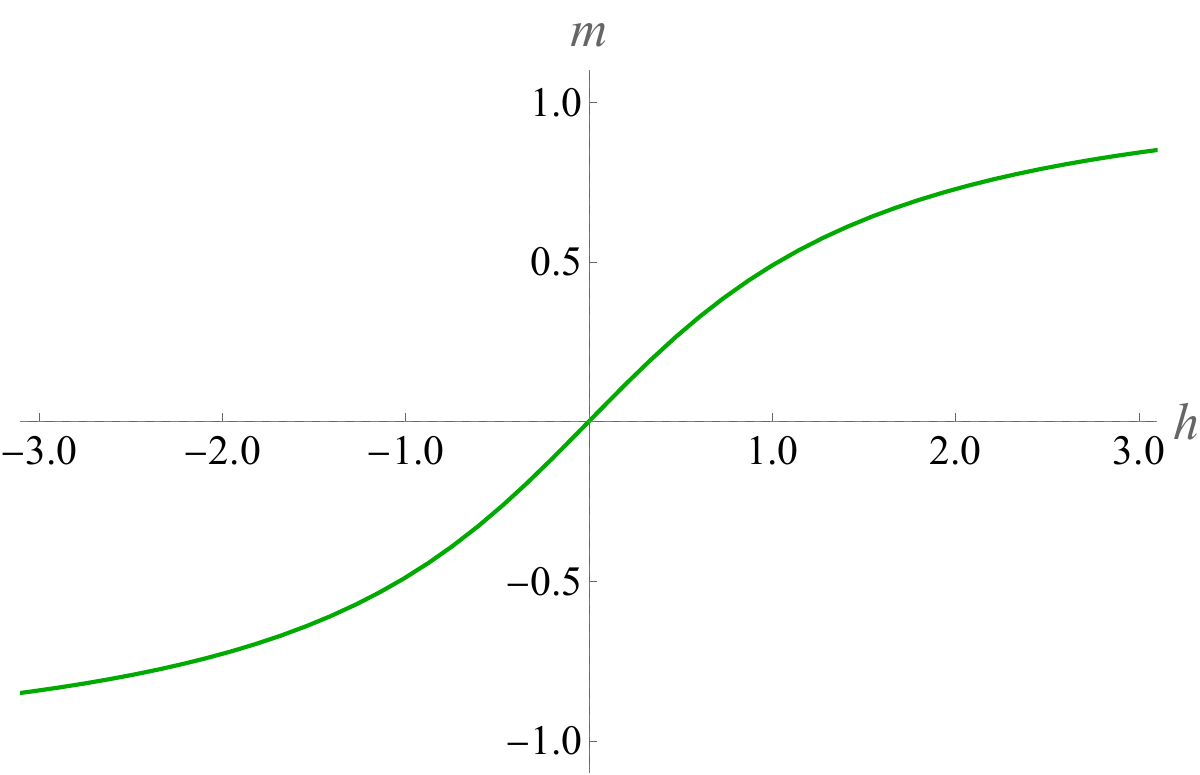}
  \includegraphics[width=0.48\linewidth]{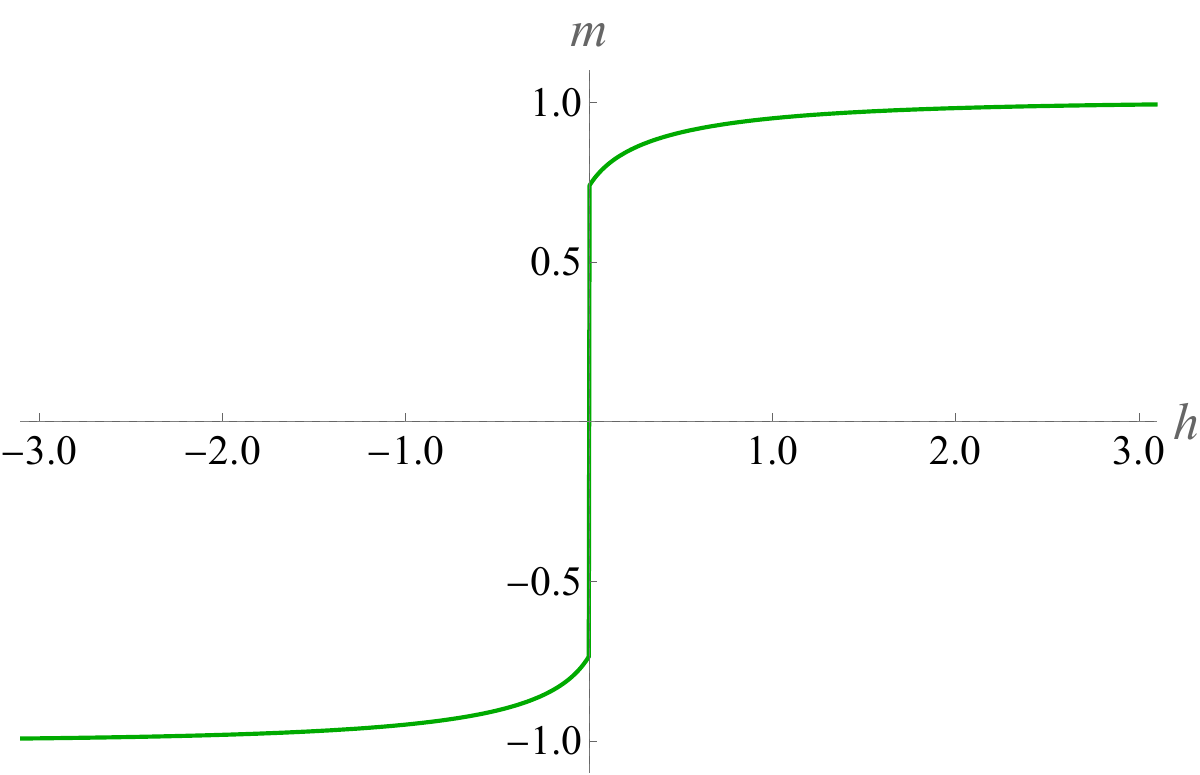}
  \caption{Magnetization $m$ versus external field $h$ for $k=4$, $J=1$ and a fixed $\beta = 0.2<\beta_c^{\rm tree} = \frac12\log2 = 0.346\dots$ (left) and $\beta = 0.4>\beta_c^{\rm tree}$ (right).}
  \label{fig:m0_vs_h}
\end{figure}

\subsection{Spontaneous magnetization}

The spontaneous magnetization is derived from the free energy \eqref{eq:free_energy_theorem} via the thermodynamic relation~\eqref{eq:def_mag}, which, after a change of coordinates, is equivalent to
\begin{equation}
  \label{eq:magnetization}
  m(\beta,h) = \beta(1-\tau^2) \frac{\partial f}{\partial \tau} = \tau + (1-\tau^2) \frac{\partial}{\partial\tau} \log V_k(\theta^\star(\kappa,\tau),\tau).
\end{equation}
Because $\theta^\star$ extremizes the potential $V_k$, the derivative $\frac{\partial V_k}{\partial \theta} \frac{\partial \theta^\star}{\partial \tau} = 0$ vanishes. Consequently, 
\[
  \frac{\partial}{\partial \tau} \log V_k(\theta^\star,\tau) = \left. \frac{\partial \log V_k(\theta,\tau)}{\partial \tau} \right|_{\theta=\theta^\star}.
\]
Evaluating this in the zero-field limit ($h \searrow 0$) yields the spontaneous magnetization
\begin{equation}
  \label{eq:M_zero_field} 
  m_0(\beta) = \left.\frac{\cos^k(\theta^\star - \phi) - \cos^k(\theta^\star + \phi)}{\cos^k(\theta^\star - \phi) + \cos^k(\theta^\star + \phi)} \right|_{h =0}\,,  
\end{equation}
where we recall that $\theta^\star$ is a function of $\kappa$ and $\tau$ and therefore 
of $\beta$ and $h$ by Notation~\ref{not:variable_changes}.
For temperatures above the phase transition point ($\kappa<\frac{1}{k-1}$), we are in case $4$ of Theorem~\ref{thm:free_energy}: $\theta=0$ is the unique global maximum of $|V_k(\theta)|$, and therefore $\theta^\star = 0$ and $m_0=0$.

For temperatures below the phase transition point ($\kappa > \kappa_c = \frac{1}{k-1}$), the global maximum $\theta^\star$ of $|V_k(\theta)|$ contributing to $f$ is determined by \eqref{eq:crit}. We introduce $u$ as defined in Corollary~\ref{cor:critical_exponents} and substitute it into \eqref{eq:crit} to obtain the critical point equation
\[
  \frac{\kappa - u}{\kappa + u} = \left( \frac{1 - u}{1 + u} \right)^{k-1}.
\]
For $\kappa > \kappa_c$, this equation admits a unique nontrivial solution $u > 0$. Once $u$ is determined, the spontaneous magnetization $m_0$ is calculated directly as
\[
  m_0(\beta) = \frac{(1+u)^k - (1-u)^k}{(1+u)^k + (1-u)^k}.
\]

\begin{proof}[Proof of Corollary~\ref{cor:phase_transitions}]
  The statement follows from Corollary~\ref{cor:smoothf} and
  the explicit computation above confirming that $f(\beta,h)$ is continuous but not differentiable when $h=0$ and $\beta > \beta_c^\mathrm{tree}$.
  Furthermore, we see that we can extend $f$ to a function that is once but not twice differentiable at $(\beta,0) = (\beta_c^{\rm tree},0)$.
\end{proof}

In Figure~\ref{fig:m0_vs_h} we depict the magnetization as a function of the external field $h$ for fixed $k=4,~J=1$ and at fixed inverse temperature $\beta$ below (left-hand side) and above (right-hand side) critical inverse temperature $\beta_c^{\rm tree}$. If $\beta > \beta_c^{\rm tree}$ we observe a discontinuity of $m$ at $h=0$, leading to a nonzero spontaneous magnetization. A plot of the spontaneous magnetization as a function of $\beta$ is shown in Figure~\ref{fig:m0_vs_temp}.

\begin{figure}
  \centering
  \includegraphics[width=0.6\linewidth]{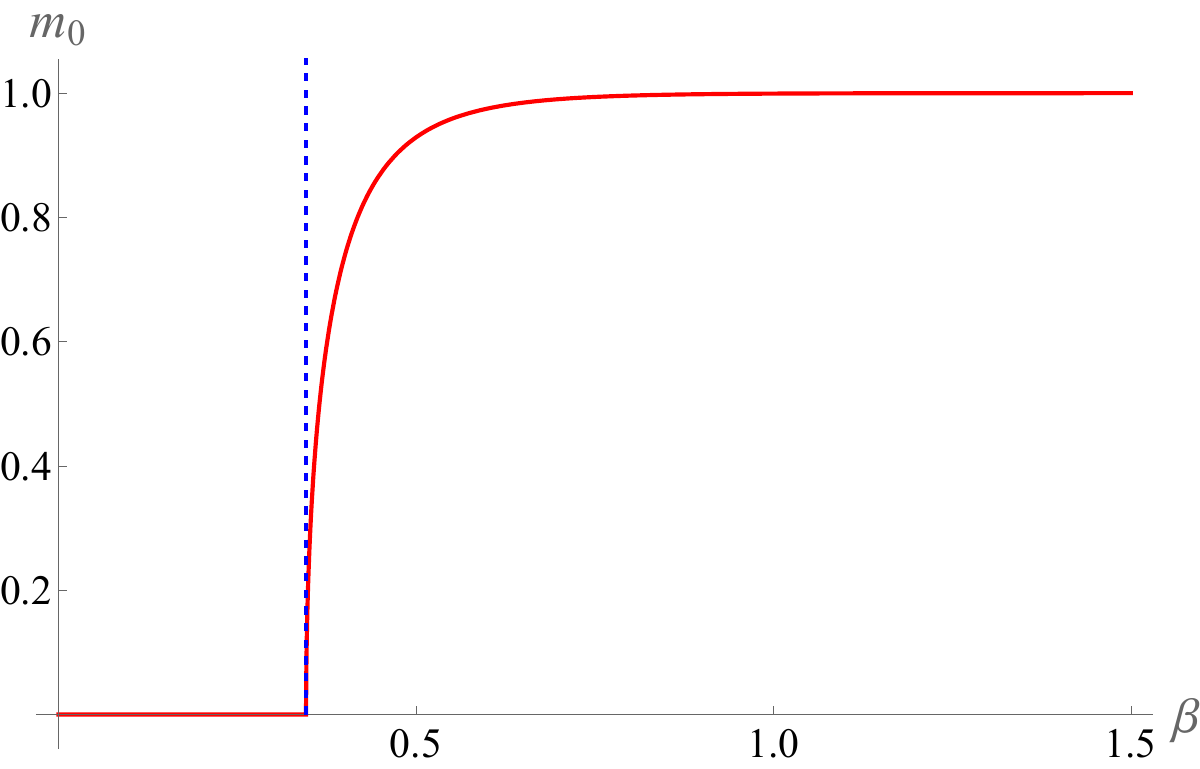}
  \caption{The curve represents the  spontaneous magnetization $m_0$ as a function of the inverse temperature $\beta$ with the parameters $k = 4$ and $J=1$. The vertical dashed line marks the second-order phase transition at $\beta_c^{\rm tree} = \frac12\log2$.}
  \label{fig:m0_vs_temp}
\end{figure}

\subsection{Critical behavior}
\label{subsec:crit_behavior}
In this section we prove Corollary~\ref{cor:critical_exponents} by computing leading-order approximations near criticality for the physical quantities defined in Section~\ref{sec:intro}. Near the phase transition the global maximum of Lemma~\ref{lem:maximum} approaches $\theta = 0$ as $(\tau, \kappa - \kappa_c) \to (0,0)$. Therefore, we compute a Taylor expansion of $V_k$ around the point $(\theta,\tau,\kappa) = (0,0,\kappa_c)$.
It is convenient to decompose $V_k(\theta, \tau, \kappa)$ into its symmetric (zero-field) and asymmetric ($\tau$-dependent) components,
\[
  V_k(\theta, \tau,\kappa) = \frac{(1+\kappa)^{k/2}}{2\cdot k!} \left( \left( \cos^k(\theta - \phi) + \cos^k(\theta + \phi) \right) + \tau \left( \cos^k(\theta - \phi) - \cos^k(\theta + \phi) \right) \right).
\]
Let $t = (T-T_c)/T_c$ denote the normalized temperature and define $\Delta(\kappa) = \kappa - \kappa_c$. Expanding $\kappa = \tanh\big(\tfrac{J}{k_B T}\big)$ around $T_c$ gives
\begin{equation}
  \label{eq:Delta_expand}
  \Delta(\kappa) = \kappa-\kappa_c = \left. t\cdot T_c \cdot\frac{\partial \kappa}{\partial T} \right|_{T=T_c} + \mathcal{O}(t^2) = \mathrm{arctanh}\left(\frac{1}{k-1}\right) \frac{k(k-2)}{(k-1)^2} (-t) + \mathcal{O}(t^2).
\end{equation}

Using the $\theta$ expansion of the asymmetric component 
\[
  \cos^k(\theta - \phi) - \cos^k(\theta + \phi) = 2k \cos^k(\phi)\tan(\phi) \, \theta + \mathcal{O}(\theta^3)
\]
and a similar expansion for the symmetric component up to quartic order, we obtain the Taylor expansion of $V_k$ near $(\theta,\tau,\Delta(\kappa)) = (0,0,0)$
\begin{equation}
  \label{eq:V_Expansion_kappa}
  V_k(\theta, \tau, \kappa) = \frac1{k!} \left( 1 + C_{2,0,1} \Delta(\kappa)\theta^2 + C_{4,0,0}\theta^4 + C_{1,1,0} \theta\tau + R(\theta,\tau,\kappa) \right),
\end{equation}
with coefficients 
\[
  C_{2,0,1} = \frac{k(k-1)}{2},\quad C_{4,0,0} = -\frac{k^2(k-2)}{12(k-1)},\quad C_{1,1,0} = k \tan(\phi_c) = \frac{k}{\sqrt{k-1}}\,,
\]
where $\phi_c = \arctan(\sqrt{\kappa_c})$ 
and error term $R(\theta,\tau,\kappa)$ satisfying
\[
  R(\theta,\tau,\kappa) = \mathcal{O}\left(
    \theta^6 + \theta^4 |\Delta(\kappa)| + \theta^2\Delta(\kappa)^2 + |\theta|^3|\tau| + 
    |\theta \tau \Delta(\kappa)|
  \right).
\]
Using \eqref{eq:Delta_expand}, we can rewrite \eqref{eq:V_Expansion_kappa} in terms of the normalized temperature as 
\begin{equation}
  \label{eq:V_Expansion}
  V_k(\theta, \tau, t) = \frac1{k!} \left( 1 + C'_{2,0,1} t \theta^2 + C_{4,0,0}\theta^4 + C_{1,1,0} \theta\tau + R'(\theta,\tau,t) \right),
\end{equation}
with
\[
  C'_{2,0,1} = -\mathrm{arctanh}\left( \frac{1}{k-1} \right) \frac{k^2 (k-2)}{2(k-1)},~ \text{and}~  R'(\theta,\tau,t) = \mathcal{O}\left(
    \theta^6 +  \theta^4 |t| + \theta^2 t^2 + | \theta|^3|\tau| + |\theta\tau t|
  \right).
\]

\subsubsection{The magnetization near critical temperature and the exponent \texorpdfstring{$\beta^{\rm mag}$}{βᵐᵃᵍ}}

The critical angle $\theta^\star$ satisfies the condition $\frac{\partial V_k}{\partial \theta}(\theta^\star,\tau,\kappa) = 0$. Using \eqref{eq:V_Expansion} and truncating yields
\[
  2 C'_{2,0,1} t \theta^\star + 4 C_{4,0,0}\cdot (\theta^\star)^3 + \mathcal{O}((\theta^\star)^5, |t|(\theta^\star)^3, t^2 \theta^\star) = 0.
\]
For $T < T_c$ (where $\Delta (\kappa) > 0$), using \eqref{eq:Delta_expand}, we obtain for the nontrivial root
\begin{equation}
  \label{eq:theta0_kappa}
  (\theta^\star)^2 = -\frac{C'_{2,0,1}t}{2 C_{4,0,0}} + \mathcal{O}(t^2) = \frac{3(k-1)^2}{k(k-2)} \Delta (\kappa) + \mathcal{O}(\Delta(\kappa)^2) = \frac{3}{2} \log\left(\frac{k}{k-2}\right) (-t) + \mathcal{O}(t^2). 
\end{equation}

To compute the spontaneous magnetization $m_0$ for small $t$, we evaluate \eqref{eq:M_zero_field} at small $\theta^\star$. Applying the first-order approximations $\cos(\phi_c \mp \theta^\star) = \cos(\phi_c) \, (1 \pm \theta^\star \tan(\phi_c)) + \mathcal{O}((\theta^\star)^2)$ and the expansion $(1+x)^k = 1+kx + \mathcal{O}(x^2)$, the magnetization simplifies to
\begin{align*}
 \begin{aligned}
    m_0(\beta) &= \frac{\cos^k(\phi_c) \left( (1 + k \theta^\star \tan(\phi_c)) - (1 - k \theta^\star \tan(\phi_c)) \right) + \mathcal{O}((\theta^\star)^3)}{\cos^k(\phi_c) \left( (1 + k \theta^\star \tan(\phi_c)) + (1 - k \theta^\star \tan(\phi_c)) \right) + \mathcal{O}((\theta^\star)^2)} \\
    &= k \tan(\phi_c) \theta^\star + \mathcal{O}((\theta^\star)^3)  \\
    &= \frac{k}{\sqrt{k-1}} \theta^\star + \mathcal{O}((-t)^{3/2})\,,
\end{aligned}
\end{align*}
where we recall that $\beta$ is a function of the normalized temperature $t$.

Substituting \eqref{eq:theta0_kappa} into this expression, we arrive at the leading-order critical behavior of the spontaneous magnetization as $t \nearrow 0$
\[
  m_0(\beta) = \frac{k}{\sqrt{k-1}} \sqrt{\frac{3}{2} \log\left(\frac{k}{k-2}\right)} \, (-t)^{1/2} + \mathcal{O}((-t)^{3/2})\,.
\]
By definition~\eqref{eq:def_beta^{mag}}, we obtain the critical exponent $\beta^{\rm mag} = \frac{1}{2}$.

\subsubsection{The specific heat critical exponent \texorpdfstring{$\alpha$}{α}}
\label{par:alpha_exponent}

Again, let us consider $h=0$. For positive normalized temperature, $t>0$, the global maximum of $|V_k(\theta)|$ is at $\theta^\star = 0$. Therefore, the high-temperature free energy function is given by 
\[
  f_{\rm ht}(t) = \frac{k}{2\beta(t)} \log(\cosh(\beta(t) J)).
\]
For $t < 0$, the maximum is given by the nonzero \eqref{eq:theta0_kappa}. Substituting the leading-order approximation \eqref{eq:theta0_kappa} for $\theta^\star$ into the expansion \eqref{eq:V_Expansion}, we obtain for $t>0$
\[
  f(-t,0) - f_{\rm ht}(-t) = -\frac{(C'_{2,0,1} t)^2}{4\beta_c^{\rm tree} C_{4,0,0}} + \mathcal{O}(t^3)=\frac{3Jk^2(k-2)\mathrm{arctanh}(\frac1{k-1})}{4(k-1)}t^2 + \mathcal{O}(t^3).
\]
Therefore, the limit \eqref{eq:singular_free_energy} is well-defined,
\[
  \lim_{t \searrow 0} \frac{\log \left| f(-t,0) - f_{\rm ht}(-t) \right|}{\log t} = 2,
\] 
and we obtain the specific heat critical exponent $\alpha=0$. In particular, this shows that the phase transition at $(\beta_c^{\rm tree},0)$ is second-order, completing the proof of Corollary~\ref{cor:phase_transitions}.

\subsubsection{The critical exponent \texorpdfstring{$\delta$}{δ}}

At critical temperature $t = 0$, the term $C'_{2,0,1} t\theta^2$ in the expansion \eqref{eq:V_Expansion} vanishes. The critical angle $\theta^\star$ is therefore determined to leading order by
\[    
  4 C_{4,0,0} \cdot (\theta^\star)^3 + \tau C_{1,1,0} = \mathcal{O}((\theta^\star)^5, |\tau|(\theta^\star)^2).
\]
Substituting $\theta^\star$ into $m(\beta,h)$ yields the leading-order behavior
\[    
  m(\beta_c^{\rm tree},h) = \frac{k}{\sqrt{k-1}} \left( \frac{3\sqrt{k-1}}{k(k-2)} \right)^{1/3} \tau^{1/3} + \mathcal{O}(\tau).
\]
Since $\tau = \tanh(\beta h)$ scales linearly in $h$ to leading order, we find $\delta=3$.

\subsubsection{Zero-field magnetic susceptibility and the exponent \texorpdfstring{$\gamma$}{γ}}

Recall that the magnetic susceptibility $\chi(\beta,h)$ is defined as the $h$-derivative of the magnetization $m(\beta,h)$. Using the identity \eqref{eq:magnetization} and assuming $t>0$ and $h=0$ (so $m(\beta,h) = 0$), we obtain
\[
  \chi(\beta,0) = \beta \left.\left(\frac{\partial}{\partial\theta} \frac{\partial}{\partial\tau} \log V_k(\theta,\tau) \right) \right|_{\theta=\theta^\star(\beta,0), \tau=0} \cdot \left. \frac{\partial \theta^\star}{\partial \tau} \right|_{\tau=0} + \O(1),
\]
where the other derivatives are absorbed in the $\O(1)$ correction. In the $t\searrow 0$ limit, the first factor above becomes $\tfrac{\beta_c^{\rm tree} k}{\sqrt{k-1}}$. Applying partial derivatives $\partial_\theta \partial_\tau$ to \eqref{eq:V_Expansion} and evaluating at $\theta=\theta^\star$ and $\tau=0$, leads to 
\[    
  \left. \frac{\partial \theta^\star}{\partial \tau} \right|_{\tau=0} = - \frac{\partial^2 V_k}{\partial \tau \partial \theta} \left(\frac{\partial^2 V_k}{\partial \theta^2}\right)^{-1} \Bigg|_{\theta=\theta^\star, \tau=0} = -\frac{C_{1,1,0}}{2C'_{2,0,1}t + 12C_{4,0,0}\cdot(\theta^\star)^2} + \O(1).
\]
Finally, using \eqref{eq:Delta_expand}, we obtain that the susceptibility scales like
\[
  \chi(\beta, 0) \sim \frac{t^{-1}}{J(k-2)} \quad\text{as } t\searrow 0,
\]
and hence we obtain the critical exponent $\gamma=1$. This completes the proof of Corollary~\ref{cor:critical_exponents}.\qed

\section{Combinatorial approach to magnetization}
\label{sec:magnetization_combinatorial}

In this section we give a combinatorial interpretation of the magnetization which directly emerges from the edge-bicolored graph counting perspective.\par

For a fixed graph $G$, the magnetization can also be defined as the average magnetic moment per vertex in the model \cite[(1.7.12)]{baxter1982exactly}, i.e.\ 
\[
  m_G(\beta,h) = \frac{1}{|V_G|} \langle \sigma_1 + \dots + \sigma_{|V_G|} \rangle_G.
\]
Since we are considering random labeled graphs, by linearity of the expectation, we may compute the expected value $\langle \sigma_1 \rangle_G$ for a random vertex $v_1$. The expectation is taken over the Gibbs probability measure, i.e.\ for any function $p\,:\, \left\{ \pm 1 \right\}^{|V_G|}\rightarrow\RR$
\[
  \langle p \rangle_G = \frac{Z_G(p,\beta,h)}{Z_G(\beta, h)},\quad \text{where } Z_G(p,\beta,h) = \!\!\!\sum_{\sigma\in \left\{ \pm 1\right\}^{|V_G|}} p(\sigma) e^{-\beta H_G}.
\]

The graphical expansion of the partition function, Proposition~\ref{prop:combinatorial_formula_partition_fct}, generalizes to a graphical expansion of the expected values $\langle \sigma_A \rangle_G$ for any $A \subseteq V_G$, see \cite[\S 2.2]{duminil2016random}. Namely,
\begin{equation*}
  \langle \sigma_A \rangle_G = \frac{\sum_{\gamma \subseteq E_G} \kappa^{|\gamma|} \tau^{|\mathfrak{o}(\gamma) \Delta A|}}{\sum_{\gamma \subseteq E_G} \kappa^{|\gamma|} \tau^{|\mathfrak{o}(\gamma)|}},
\end{equation*}
where, as before, $\kappa = \tanh(\beta J)$, $\tau = \tanh(\beta h)$, $\Delta$ denotes the symmetric difference, and 
\[
  \mathfrak{o}(\gamma) = \{v \in V_G \,:\, v \text{ has odd degree in the subgraph induced by } \gamma\}.
\]
To pass to the random graph model, we define the weighted (annealed) magnetization
\begin{equation}
  \label{eq:weighted_magnetization_finite}
  \widetilde{m}_{\epsilon n,k}(\beta, h) := \frac{\frac{1}{\epsilon n} \sum_{G\in \GG_{\epsilon n,k}^{\bullet}} \frac{1}{|\Aut(G)|} Z_G(\sigma_1,\beta,h)}{\sum_{G\in \GG_{\epsilon n,k}} \frac{1}{|\Aut(G)|} Z_G(\beta,h)} \quad \text{with } \epsilon = 1 + (k ~\mathrm{mod}~ 2), 
\end{equation}
where $\GG_{\epsilon n,k}^{\bullet}$ denotes the set of $k$-regular graphs with a single marked vertex (referred to as $v_1$) and $\epsilon n$ many vertices. Here, the $\epsilon$ ensures that the denominator does not vanish, since an odd regular graph needs to have an even number of vertices. Automorphisms in the numerator are those preserving the marked vertex.
The denominator is the $k$-regular random graph partition function from before. To compute the numerator of \eqref{eq:weighted_magnetization_finite}, we need to enumerate edge-bicolored graphs where exactly one (marked) vertex $v_1$ contributes $\lambda_{v_1}^{(G,\gamma)}$ with flipped parity compared to Proposition~\ref{prop:combinatorial_formula_partition_fct}, viz.\
\begin{equation}
  \label{eq:vtx_weight_flipped}
  \lambda_{v_1}^{(G,\gamma)} = \begin{cases}
    \omega_{v_1}^{(G,\gamma)} \cdot \tau & \text{if } \deg_{\gamma}(v_1) \text{ is even,} \\
    \omega_{v_1}^{(G,\gamma)} & \text{if } \deg_{\gamma}(v_1) \text{ is odd.}
  \end{cases}
\end{equation} 

We explain how to derive the asymptotics of the numerator in the large-$n$ limit. 
Let $A_{n,k}$ be the weighted number of edge-bicolored $k$-regular graphs with $n$ vertices, where each vertex $v$ is weighted by $\lambda_{v}^{(G,\gamma)}$ as in \eqref{eq:comb_vtx_weights}, except for a single marked vertex $v_1$ whose weight is given by \eqref{eq:vtx_weight_flipped}.
The generating function for $A_{n,k}$ is given by
\begin{equation}
  \label{eq:gen_fun_one_vtx_flipped}
  \sum_{n\geq 0}A_{n,k}(\kappa,\tau)z^n =  \sum_{s,t \geq 0} (2s-1)!! (2t-1)!! [x^{2s} y^{2t}] \left( z w_k(x,y) \exp(z V_k(x,y)) \right),
\end{equation}
where we define
\[
  w_k(x,y) := \sum_{\ell=0}^k \kappa^{\ell/2} \tau^{(\ell+1)\,\mathrm{mod}\, 2} \frac{x^{k-\ell}y^\ell}{\ell!(k-\ell)!},
\]
and, as before, we have the potential
\[
  V_k(x,y) =  \sum_{\ell=0}^k \kappa^{\ell/2} \tau^{\ell~\mathrm{mod}\, 2} \frac{x^{k-\ell}y^\ell}{\ell!(k-\ell)!}.
\]
A proof of \eqref{eq:gen_fun_one_vtx_flipped} is analogous to the proof of \cite[Proposition 3.5]{BMW_bicolored}. Extracting the $n$th coefficient via a Cauchy integral, one obtains along similar lines of \cite[Lemma 3.1]{BMW_multicolor}
\[
  A_{n,k} = \frac{2^{kn/2} \Gamma\big(\frac{kn+2}{2}\big)}{2\pi (n-1)!} \int_{S^1} w_k(\cos(\theta),\sin(\theta)) V_k(\cos(\theta),\sin(\theta))^{n-1} \mathrm{d}{\theta}.
\]
This integral is amenable to stationary phase approximation, see \cite[\S 5]{deBruijnAsymptotic}, and we get
\[
  A_{n,k} \underset{n\to\infty}{\sim} \left\{\begin{array}{ll}
    \frac{n}{2\sqrt{2}\pi} k^{\frac{nk+ 1}{2}} \left( \frac{2}{k-2} \right)^{\ell-\frac{1}{2}} \Gamma(\ell) \sum_{(x,y) \in \Phi(\beta,h)} \frac{w_k(x,y) V_k(x,y)^{n-1}}{\sqrt{B_k(x,y)}} & \text{if } nk \text{ is even}, \\
    0 & \text{otherwise},
  \end{array}\right. 
\]
where $\ell = \frac{k-2}{2}n$ and $\Phi(\beta, h)$ as well as $B_k$ are exactly as in Proposition~\ref{prop:asymptotic_limit_Z_n,k}. Here, we assume that all critical points in $\Phi$ are nondegenerate. Since, by Theorem~\ref{thm:free_energy}, there is a unique (up to antipodal symmetry) global maximum of $|V_k(\theta)|$ that contributes to $\Phi(\beta,h)$ (possibly with multiplicity), we obtain in the large-$n$ limit
\begin{equation}
  \label{eq:mtilde_asymp}
  \widetilde{m}_k := \lim_{n\to\infty} \widetilde{m}_{\epsilon n,k}(\beta, h) = \frac{w_k(\theta^\star)}{V_k(\theta^\star)} \quad \text{with } \epsilon = 1 + (k ~\mathrm{mod}~ 2),
\end{equation}
where $\theta^\star$ is the global maximizer of $|V_k(\theta)|$ exactly as in Theorem~\ref{thm:free_energy}. Using the identity
\[
  w_k = \tau V_k + (1-\tau^2) \frac{\partial}{\partial \tau} V_k,
\]
and comparing \eqref{eq:mtilde_asymp} with \eqref{eq:magnetization}, we see that asymptotically, in the thermodynamic limit, up to leading order $\widetilde{m}_k$ agrees with the magnetization $m(\beta,h)$ from before.

\printbibliography

\noindent 
\small
\textsc{Michael Borinsky}\\
\textsc{Perimeter Institute, 31 Caroline St N, Waterloo, Ontario N2L 2Y5, Canada}\\
\url{mborinsky@perimeterinstitute.ca}\medskip\\
\noindent
\textsc{Shiyue Ren}\\
\textsc{Perimeter Institute, 31 Caroline St N, Waterloo, Ontario N2L 2Y5, Canada}\\
\textsc{Department of Physics and Astronomy, University of Waterloo, 200 University Avenue West, Waterloo, Ontario N2L 3G1, Canada}\\
\url{rshiyue@uwaterloo.ca}\medskip\\
\noindent
\textsc{Maximilian Wiesmann}\\
\textsc{Center for Systems Biology Dresden}\\
\textsc{Max Planck Institute of Molecular Cell Biology and Genetics}\\
\textsc{Max Planck Institute for the Physics of Complex Systems}\\
\textsc{Technische Universität Dresden}\\
\url{wiesmann@pks.mpg.de}

\end{document}